\documentclass[letterpaper, 10 pt, conference]{ieeeconf}  

\usepackage{amssymb}
\usepackage{stmaryrd}
\usepackage{cite}
\usepackage{paralist}
\usepackage{amsmath}
\usepackage{amsthm}
\usepackage{amsfonts}
\usepackage{graphicx}
\usepackage{epstopdf}
\usepackage{caption}
\usepackage[subrefformat=parens,labelformat=parens]{subfig}
\usepackage{listings}
\usepackage{paralist}
\usepackage{tikz}
\usepackage[ruled,commentsnumbered,linesnumbered]{algorithm2e}
\usepackage{lipsum}
\usepackage{hyperref}
\usepackage{url}
\usepackage{multicol}
\usepackage{blindtext}
\usepackage{array}
\usepackage{floatrow}
\usepackage{caption}
\usepackage{lipsum}
\usepackage{wrapfig}
\usepackage{graphicx}
\usepackage{mathabx}

\newcommand{\loc}[0]{\mathit{loc}}

\newcommand{\safe}[0]{\mathit{Safe}}

\newcommand{\init}[0]{\mathit{Init}}

\newcommand{\goal}[0]{\mathit{Goal}}

\newcommand{\must}[0]{\mathit{Must}}
\newcommand{\may}[0]{\mathit{May}}
\newcommand{\wpst}[3]{ {#3}\text{-}\overline{post}^{#1}(#2,\vu)}
\newcommand{\npst}[3]{ {#3}\text{-}\underline{post}^{#1}(#2,\vu)}
\newcommand{\pst}[2]{post^{#1}(#2,\vu)}

\newcommand{\synprob}[0]{\mathit{RAC}}

\newcommand{\CPar}[0]{controller\ table}
\newcommand{\solveorigin}[1]{$#1 \models \Pi(M,\C,\V)$}
\newcommand{\solveover}[1]{$#1 \models \Pi^s_\P(M,\C,\V)$}
\newcommand{\solveunder}[1]{$#1 \models \Pi^w_\P(M,\C,\V)$}

\newcommand{\solveoriginprime}[2]{$#1 \models \Pi(M#2,\C,\V)$}

\newcounter{Theorem}
\newtheorem{definition}[Theorem]{Definition}

\newtheorem{prop}[Theorem]{Proposition}
\newtheorem{lemma}[Theorem]{Lemma}

\newtheorem{theorem}[Theorem]{Theorem}

\newcommand{\final}[1]{{ #1 }}

\newcommand{\num}[1]{\relax\ifmmode \mathbb #1\else $\mathbb #1$\fi}
\newcommand{\reals}{{\num R}}                    
\newcommand{\naturals}{{\num N}}                      

\newcommand{\vu}{{\bf u}}

\def\B{{\cal B}} 
\def\C{{\cal C}} 
\def\F{{\cal F}} 
\def\I{{\cal I}} 
\def\M{{\cal M}} 
\def\N{{\cal N}}
\def\P{{\cal P}} 
\def\V{{\cal V}} 
\def\U{{\cal U}} 
\def\X{{\cal X}} 

\begin{document}

\title{\LARGE \bf
Controller Synthesis with Inductive Proofs for Piecewise Linear Systems: an SMT-based Algorithm
}

\author{%
Zhenqi Huang, Yu Wang, Sayan Mitra, Geir E. Dullerud and  Swarat Chaudhuri
\thanks{%
Zhenqi Huang, Yu Wang, Sayan Mitra and Geir E. Dullerud are with Coordinated Science Laboratory, University of Illinois at Urbana-Champaign, Urbana, IL 61801, USA
{\tt\small \{zhuang25, yuwang8, mitras, dullerud\}@illinois.edu}}
\thanks{Swarat Chaudhuri is with the Department of Computer Science,  Rice University, Houston, TX 77005, USA
{\tt\small swarat@rice.edu}}
}

\maketitle
\thispagestyle{empty}
\pagestyle{empty}

\begin{abstract}
We present a  controller synthesis algorithm for reach-avoid problems for piecewise linear discrete-time systems. Our algorithm relies  on SMT solvers and in this paper we focus on piecewise constant control strategies. Our algorithm generates feedback control laws together with inductive proofs of unbounded time safety and progress properties with respect to the reach-avoid sets. Under a reasonable robustness assumption, the algorithm is shown to be complete. That is, it either generates a controller of the above type along with a proof of correctness, or it establishes the  impossibility of the existence of such controllers. To achieve this, the algorithm iteratively attempts to solve a weakened and strengthened versions of the SMT encoding of the reach-avoid problem. We present preliminary experimental results on applying this algorithm based on a prototype implementation. 
\end{abstract}

\section{Introduction}
\label{sec:intro}
A Satisfiability Modulo Theory (SMT) problem is a classical decision problem in computer science~\cite{CDP97}. It takes as input a logical formula in first-order logic that can involve combinations of background theories, and requires one to decide whether or not the formula has a satisfying solution. 
For a bounded time horizon $k$, a simplest SMT problem in Equation~(\ref{eq:exsmt}), for instance, is an encoding of a  search for a sequence of control inputs vectors $u_1, \ldots, u_k$ that drives a discrete time linear open-loop control system from every initial state in the hypercube  $[0,0.1]^n$ to the hypercube $[0.9,1]^n$ in $k$ steps, while always keeping the state inside the hypercube  $[0,1]^n$.
\begin{align}
& \exists \ u_1, \ldots u_k, \forall \ x_0 \in [0,0.1]^n, \forall \ t \in \{1,\ldots,k-1\},  \nonumber \\
& \ \ \mbox{let} \ x_{t+1} = A x_t + B u_t \nonumber \\ 
& \ \ \mbox{such that} \ x_t \in [0,1]^n \ \mbox{and} \  x_k \in [0.9,1]^n. \label{eq:exsmt}
\end{align}
This example has several constraints that are defined in terms of the quantified variables $u_i$ and  $x_i$, numerical constants (including those in the matrices $A$ and $B$), and the background theory of linear real arithmetic. 
An SMT solver is a software tool that solves SMT problems by  either giving an assignment to the variables that satisfy all the constraints or by saying that none exists. 
Modern SMT solvers routinely handle linear problems with thousands of variables and millions of constraints, so much as they have become the engines for innovation in verification and synthesis for computer software and hardware~\cite{de2008z3,barrett2011cvc4,dutertre2006yices}. 
Although many control systems can only be modeled by means of nonlinear arithmetic over the real numbers involving transcendental functions that make the corresponding SMT problems undecidable, the solvers are evolving rapidly and several incorporate approximate decision procedures for nonlinear arithmetic~\cite{gao_dreal:_2013}.
These technological developments motivated us (and others~\cite{saha_automated_2014,nedunuri_smt-based_2014}) to explore SMT-based controller synthesis. 

In this paper, we present an algorithm that uses SMT solvers for synthesizing controllers for discrete time systems. The dynamics of the system is given as a piecewise linear feedback control system.
%
The control requirements are the standard {\em reach-avoid\/} specification~\cite{cardenas2008research,cardenas2009challenges}: a set of states $\goal$ that has to be reached while always staying inside a  $\safe$ set.

A key  difficulty in using SMT for synthesis, is that the resulting SMT problem has to encode the unrolled dynamics of the system with the unknown controller inputs. In the above simple example, this gave rise to $k$ control input variables and the intermediate states. For more general nonlinear models, the intermediate states cannot be written down in closed form and one has to unroll the over-approximations of the dynamics. This can then lead to overtly conservative answers from the solver.  
We present a technique that avoids this problem by synthesizing the control law together with an inductive proof of its correctness. The proof has two parts: (a) an inductive invariant that implies safety and (b) a ranking function that implies progress. A positive side-effect of this is that it can not only synthesize controllers with understandable correctness proofs, but it can also establish the nonexistence of provably correct controllers (of a certain template). 

In Section~\ref{sec:prelim} we define the system model,  the reach-avoid synthesis problem and a particular notion of robustness of models.
In Section~\ref{sec:rules} we first present a basic SMT encoding of the synthesis problem and then a strengthened and a weakened version this encoding. 
Using these two encodings, in Section~\ref{sec:algorithm} we present the synthesis algorithm, its soundness and relative completeness. 
In Section~\ref{sec:case} we illustrate an application of the algorithm in a vehicle navigation problem and conclude in Section~\ref{sec:conc}.

\section{Related Work}
\label{sec:related}

Researchers have recently used SMT solvers for synthesizing programs and strategies in games. 
The approach in~\cite{saha_automated_2014} uses SMT to find controllers for general linear temporal logic (LTL) specifications 
by stitching together motion primitives from a library. Unlike our encoding with inductive proofs, the approach of~\cite{saha_automated_2014} involves bounded unrolling of the dynamics.

In~\cite{nedunuri_smt-based_2014},
the authors used SMT solvers to synthesize integrated task and motion plans by constructing a placement graph. In~\cite{beyene2014constraint}, a constraint-based approach was developed to solve games on infinite graphs between the system and the adversary. Our work can be seen as  introducing  control theoretic constraints to the SMT formulation. 


The authors of~\cite{zhou2012general,ding2011reachability} proposed a game theoretical approach to synthesize controller for the reach-avoid problem, first for continuous and later for switched systems. In these approaches, the reach set of the system is computed by solving a non-linear Hamilton-Jacobi-Isaacs PDE. 
Our methodology, instead of formulating a general optimization problem for which the solution may not be easily computable,
solves a special case exactly and efficiently.
With this building block, we are able to solve more general problems through abstraction and refinement.

\final{
Model predictive control (MPC) can as well be used to solve the reach-avoid problem~\cite{bemporad2002model,turri2013linear}.
In each cycle of an MPC, the optimal input for
reaching the goal while avoiding the obstacle, is computed for a fixed prediction horizon. Then, the first part of the optimal control input is applied, and a new input is computed from the new state, and so on. 
As the prediction horizon increases, the applied input converges to the optimal reach-avoid input. 
In the contrast, our approach can be used to synthesis controls for possibly unbounded horizon with safety and progress guarantees and can establish nonexistence of controller of certain type.
}

There is a large body of results on automata theoretic approaches for controller synthesis~\cite{svorenova_temporal_2014,aydin_gol_finite_2014,kress-gazit_temporal-logic-based_2009,wongpiromsarn_receding_2012,liu_abstraction_2014}. The approach here is to construct a finite abstraction of the dynamical system and then invoke the LTL synthesis algorithms such as the one in~\cite{Bloem2012911}. This approach has been applied to several systems and several software tools for synthesis have been implemented~\cite{chen2006,mazo2010pessoa}.

The authors of~\cite{lahijanian_temporal_2012,svorenova_temporal_2015} 
build Markov decision trees to synthesize control policies with maximum probability 
of satisfying the specifications. Our method is very different since we consider deterministic systems and try to synthesize controller that are guaranteed to satisfy the specifications.


\section{preliminaries and Background}
\label{sec:prelim}

\paragraph*{Sets and Functions} 
For a natural number $N$, $[N]$ denotes the set $\{0,1,\ldots,N-1\}$.
Given two functions $f,g:A\rightarrow \reals^n$, we use $d(f,g) = |f(a)-g(a)|_\infty$ to denote the $\ell_\infty$ distance
between $f$ and $g$, where $|\cdot|_\infty$ is the standard $\infty$-norm.

We will use finite collections of sets to approximate arbitrary compact subsets in $\reals^n$.  For a finite collection $\P$ of  subsets of $\reals^n$ and a subset $S\subseteq \reals^n$, we say that $\P$ {\em preserves} $S$ if
there exists a subset $\P'\subseteq \P$ such that
(i) $\bigcup\limits_{P\in \P'} P = S$, and
(ii) $\forall \ P \in\P\backslash \P', P \cap S = \emptyset$.
In other words, $\P'$ completely and exactly represents $S$.
%

A {\em finite partition} $\P$ of a compact subset $S\subseteq \reals^n$ is a finite disjoint collection of sets that exactly cover $S$. The {\em resolution\/} of  a partition $\P$ is the maximum diameter of the sets in $\P$.
For two partitions $\P,\P'$ of a compact set $S$, we say that $\P$ {\em subsumes\/} $\P'$, if for any $I\in \P$, there exists $I'\in \P'$ such that $I\subseteq I'$.


%

 
\paragraph*{Piecewise Linear Systems, Feedback, and Robustness}
\label{sec:model}
A piecewise linear system $M$ is a tuple $(\X,\U,\loc,\I,\F)$
where 
(a) $\X\subseteq \reals^n$ is a compact set called the  {\em state space}, 
(b) $ \U \subseteq \reals^m$ is a compact set called the {\em input space}, 
(c) $\loc$ is a finite set  called the {\em set of locations}, 
(d) $\I = \{P_l\}_{l \in \loc}$ is a  partition of $\X$ and each element of $\I$ is called a  
{\em location invariant\/}, and
(e) $\F= \{f_l\}_{l \in \loc}$ is a collection of linear {\em dynamic functions}
$f_l: \X \times \U \rightarrow \X$.

The evolution of the continuous state of the system is governed by the dynamic function of the location invariant it is currently in.
For any time $t\in \naturals$, a state $x_t\in P_l$ and an input $u_t\in \U$ 
the next state of the system is given by the discrete-time dynamics: 
\begin{equation}
\label{eq:dynamics}
x_{t+1} = f_l(x_t,u_t).
\end{equation}

A general static state-feedback control law can be thought of as a function $\vu : \X\rightarrow \U$ that maps each state to an input. In many systems, sensors and controller hardware have a finite resolution, and therefore, such a general law cannot be implemented. 
In this paper, we assume that $M$ is associated with a   {\em \CPar} $\C$
which is a partition of the state space $\X$ and the  $\vu:\C \rightarrow \U$ maps 
each partition in $\C$ to an input. Essentially $\vu$ is a look-up table, which assign an input for every equivalence class defined by $\C$. 

For a feedback control policy $\vu$ and a system~\eqref{eq:dynamics},
the next state is just a function of the current state.
We denote
$post_M({x},\vu) = f_l(x,\vu(x))$ if $x\in P_l$. 
The subscript $M$ is dropped if it is clear in the context.
We denote by $\pst{t}{x}$ the state reached from $x$ after the $t^{th}$ step.
For a compact set of states $S\subseteq \X$,
we  define  $\pst{}{S} = \{x': \exists x\in S \text{ such that } x'=\pst{}{x}\}$.
The $t$ step post operation $\pst{t}{S}$ is defined similarly. 

Our synthesis algorithm will be complete for system models upto some imprecision in the model. 
For a system $M = (\X,\U,\loc,\I,\F)$ and $\epsilon>0$,
another system $M'$ is an {\em $\epsilon$-perturbation\/} of $M$ if it is identical to $M$
except that the set of dynamic functions for $M'$ is $\F'=\{f'_l\}_{l\in\loc}$, such that  
for each $l \in \loc$,  $d(f_l,f_l')\leq \epsilon$.
We denote by $\B_\epsilon(M)$ the set of all models that are $\epsilon$-perturbations of $M$.



\paragraph*{Reach-Avoid Control Problem ($\synprob$)}
A {\em reach-avoid control\/}($\synprob$) problem  
is parameterized by the system model $M$, the \CPar{} $\C$, and three sets of states $\init, \safe, \goal \subseteq \X$ called the {\em initial, safe} and {\em goal\/} states. 
We will assume that these sets have some finite representation (for example, hyperrectangles, polytopes). 
%
We define what it means to solve a $\synprob$ problem with a feedback control policy $\vu$.
\begin{definition} 
\label{sec:def}
A {\em solution\/} to a $\synprob$ is a feedback control policy $\vu:\C\rightarrow\U$ such that 
for any  initial state $x \in \init$, the states visited by the system satisfies the condition:
\begin{itemize}
\item {\em (Safety)\/} for all $t \in \naturals$, $x_t \in \safe$ and
\item {\em (Progress)\/} there exists $T\in \naturals$ such that $x_T \in \goal$.
\end{itemize} 
\end{definition}
Throughout the paper a $\synprob$ is uniquely specified by a model $M$  as the rest parameters are fixed.
%
%


\section{Constraint-based Synthesis}
\label{sec:rules}
A major barrier in encoding $\synprob$ as an SMT problem is that the safety and progress requirements are over unbounded time. Moreover, these requirements are stated in terms of the future reachable states of the system and computing that in and on itself is a hard problem. Instead of working with unbounded time reach sets, we address this problem by encoding a set of rules that inductively prove safety and progress of the control system.

\subsection{Inductive Synthesis Rules}
\label{ssec:indrules}

In addition to searching for the feedback control law $\vu:\C\rightarrow \U$,
the SMT problem will encode the search for 
(a) an inductive invariant $Inv \subseteq \X$ that proves safety with $\vu$, and 
(b) a ranking function $\V$ that proves progress with $\vu$. 

In order to constrain the search, we will fix a {\em template\/} for the ranking function.
For this paper, we will use the template $\C\rightarrow \naturals$, that is, any 
function that is piecewise constant on the partition of the state space $\C$. 
This choice has an easy interpretation: each entry in the \CPar{} gives the rank of the controller along with the feedback law.
Let $\V$ denote the countable set of all such functions. 
Each ranking function $V\in\V$ maps every state in $\X$ to a natural number.
For any $C\in \C$, $V(C)$ is the natural number that all the states $x \in C$ map to.
%
Now we are ready to present the basic rules encoding inductive synthesis of $\synprob$:
\begin{figure}[H]
\fbox{\parbox{.9\textwidth}{
\noindent
Find $\vu:\C\rightarrow \U$,  $V\in \V$, $Inv\subseteq \X$ such that: \\
\begin{tabular}{llcl}
R1:&$\init \subseteq Inv$ \\
R2:&$\pst{}{Inv} \subseteq Inv$ \\
R3:&$Inv \subseteq \safe$ \\
R4:&$C\subseteq \goal  \Longleftrightarrow  V(C)=0 $\\
R5:&$C\subseteq Inv \wedge \pst{}{C}\cap C' \neq \emptyset $\\
  & $\quad \Rightarrow V(C) \geq V(C')$\\
R6:&$C\subseteq Inv\backslash \goal  \wedge \pst{k}{C} \cap C'\neq \emptyset$ \\ 
 &  $\quad \implies V(C)>V(C')  $.
\end{tabular}
}
}
\caption{Basic rules $\Pi(M,\C,\V)$ for synthesis for $\synprob$.}
\label{fig:rules}
\end{figure} 
%
Rules R1-R3 imply that $Inv$ is a fixed-point of $post$ with control $\vu$ that contains 
$\init$ and is contained in $\safe$, and therefore, is adequate for proving safety. 
Rule R4 states the the rank of any region $C$ vanishes iff it is in $\goal$. 
Rule R5 encodes  the (Lyapunov-like) property that the rank of any region $C$ is nonincreasing along trajectories. 
Finally, rule R6 states that for any non-$\goal$ region $C$, the rank decreases with $\vu$ within $k$ steps, 
where $k$ is an induction parameter of this encoding.
For $\synprob$ specified by model $M$ with $\CPar{}$ $\C$ and a template $\V$, we denote the SMT problem  (Figure~\ref{fig:rules}) as $\Pi(M,\C,\V)$
For some control  $\vu$, ranking function $V\in \V$ and $Inv\subseteq \X$, 
we write \solveorigin{\vu,V,Inv}\footnote{For the sake of clarity, we supress the dependence on $k$.} if the Rules R1-R6 are satisfied.




\begin{theorem}[Soundness]
\label{thm:soundrules}
If \solveorigin{\vu,V,Inv}, then $\vu$ solves $\synprob$ specified by $M$.
\end{theorem}
\begin{proof}
Let $\vu,V,Inv$ satisfy rules in Figure~\eqref{fig:rules}.
Fixing any $x\in \init$, we prove safety and progress conditions separately.

From R1, $x\in Inv$. Combined with R2, we have for any 
$t\in \naturals$, $post_\vu(x,t)\in Inv$.
Since $Inv\subseteq \safe$ (R3) we have $post_\vu(x,t)\in \safe$ for any $t$. Thus the safety condition holds.

We assume $x\in C$ such that $C\cap\goal = \emptyset$; otherwise the progress condition holds trivially.
From R4 we have $V(C)>0$. From R5 and R6,  in at most $kV(C)$ steps,
$V$ decreases to 0. By $R4$ this implies that $x$ reaches the goal.
\end{proof}

\paragraph{Robustness Modulo Templates}

In Section~\ref{sec:model} we defined perturbations of system models,
here we lift the definition to the corresponding synthesis rules:
$\Pi(M',\C',\V')$ is an {\em $\epsilon$-perturbation} of $\Pi(M,\C,\V)$ 
if (i) the $\CPar{}$ and the ranking templates are identical $\C=\C', \V=\V'$,  and 
(ii) the model $M'\in \B_\epsilon(M)$ is an $\epsilon$-perturbation of $M$.


\begin{definition}
\label{def:robust}
For a $\CPar{}$ $\C$  and a template of ranking functions $\V$, a $\synprob$ specified by $M$ is {\em robust} modulo $(\C,\V)$ 
if there exists $\epsilon>0$ such that either of the following holds:
\begin{enumerate}[(i)]
\item there exists a control $\vu$ and a ranking function $V\in\V$ such that for any $M'\in \B_\epsilon(M)$, 
\solveoriginprime{\vu,V,Inv}{'} with some $Inv\subseteq \X$, or
\item for none of $M'\in \B_\epsilon(M)$, the synthesis problem $\Pi(M',\C,\V)$ is satisfiable.
\end{enumerate}
\end{definition}
In Theorem~\ref{thm:main} we will show that  our synthesis algorithm is also relatively complete with respect to this notion of robustness.


%
%
%

\subsection{Weakened and Strengthened Rules}
\label{ssec:relaxed} 

The main challenge in solving $\Pi(M,\C,\V)$ is the $post$ operator in Rules R2, R5, and R6. 
We need a reasonable representation of $post$ for this computation to be effective. 
In this work,  we use a finite partition $\P$ of the state space (which preserves $\C,\init,\safe,\goal$)
for computing the $post$. This choices is somewhat independent of the rest of the methodology
and any other template (for example, linear functions, support functions, zonotopes) could be used instead of the fixed partitions. 



The key idea to solve it is to create  {\em over} and  {\em under} approximations of the $post$ operator with respect to the representation of choice---in this case representation using the fixed partition $\P$.
These operators are then used to create weakened and strengthened versions of the basic inductive rules that can be effectively solved as SMT problems. 

We define an over-approximation ($\P\text{-}\overline{post}$) and an under-approximation  ($\P\text{-}\underline{post}$) of the $post$ operator with respect to a partition $\P$ as follows: for any compact $S\subseteq \X$,
\begin{eqnarray}
\label{eq:wpst}
\wpst{}{S}\P &=& \bigcup\limits_{P\in \P \wedge P\cap \pst{}{S} \neq \emptyset}P,\\
\label{eq:npst}
\npst{}{S}\P &=& \bigcup\limits_{P\in \P \wedge P\subseteq \pst{}{S}}P.
\end{eqnarray}
Roughly, the over-approximation  $\wpst{}{S}\P$ computes the minimum superset of $S$ which is preserved by $\P$ and the under-approximation  $\npst{}{S}\P$ computes the maximum subset of $S$ which is preserved by $\P$.  
We define $\wpst{t}{S}\P$ and $\npst{t}{S}\P$ as the $t$ step over and under-approximations in the usual way.

\begin{prop}
\label{prop:complete}
For any measurable $S\subseteq \X$, a $post$ operator and any partition $\P$, the following properties hold:
\begin{enumerate}[(i)]
\item $\npst{}{S}\P \subseteq \pst{}{S} \subseteq \wpst{}{S}\P$, 

\item If $\P'$ subsumes $\P$, then $\npst{}{S}{\P'}\supseteq \npst{}{S}{\P}$ and $\wpst{}{S}{\P'}\subseteq \wpst{}{S}{\P}$, and

\item For any $\epsilon>0$,  
$\exists \ \delta>0$ such that for any $\P$ with resolution less than $\delta$,
$d(\npst{}{S}\P,\wpst{}{S}\P)<\epsilon$.
\end{enumerate}
\end{prop}

Instead of searching for an exact inductive invariant $Inv$, the weakened and strengthened versions of the synthesis rules presented below try to find under ($Must$) and over-approximations ($May$) of the invariant using the $\wpst{}{S}\P$ and $\npst{}{S}\P$ operators. 
\begin{figure}[ht]
\fbox{\parbox{0.9\textwidth}{

Find $\vu:\C\rightarrow \U$,  $V\in \V$, $\must\subseteq \X$ such that   \\
\begin{tabular}{ll}
W1:& $ \init \subseteq \must $\\
W2:& $ \npst{}{\must}\P \subseteq \must $\\
W3:&	$ \must\subseteq \safe$ \\
W4:&	$C\subseteq \goal  \Longleftrightarrow  V(C)=0 $\\
W5:& $C\subseteq \must \wedge C'\subseteq \npst{}{C}\P $\\
& $\ \implies V(C) \geq V(C')$\\
W6:&	$C\subseteq \must\backslash \goal  \wedge C'\subseteq \npst{k}{C}\P$\\
  	&  $\ \implies V(C)>V(C') $
\end{tabular}
}}
\caption{ Weakened rules $\Pi_\P^w(M,\C,\V)$ for synthesis.}
\label{fig:underrules}
\end{figure}
\begin{figure}[ht]
\fbox{\parbox{.9\textwidth}{
Find $\vu:\C\rightarrow \U$,  $V\in \V$, $\may\subseteq \X$ s.t.: \\
\begin{tabular}{ll}
S1:& $ \init \subseteq \may $\\
S2:& $ \wpst{}{\may}\P \subseteq \may $\\
S3:&	$ \may\subseteq \safe$ \\
S4:&	$C\subseteq \goal  \Longleftrightarrow  V(C)=0 $\\
S5:& $C\subseteq \may \wedge C'\subseteq \wpst{}{C}\P $\\
& $\ \implies V(C) \geq V(C')$\\
S6:&	$C\subseteq \may\backslash \goal \wedge  C'\subseteq \wpst{k}{C}\P$\\
  	&  $\ \implies V(C)>V(C') $
\end{tabular}
}}
\caption{Strengthened rules $\Pi_\P^s(M,\C,\V)$ for synthesis.}
\label{fig:overrules}
\end{figure}

\begin{lemma}
\label{lem:sound}
For any $\P$, the  following hold:
\begin{enumerate}[(i)] 
\item if \solveorigin{\vu,V,Inv}, then there exist $\must\subseteq \X$ such that \solveunder{\vu,V,\must}; 
and 
\item if \solveover{\vu,V,\may}, then exists $Inv\subseteq \X$ such that \solveorigin{\vu,V,Inv}.
\end{enumerate}
\end{lemma}
\begin{proof}
Suppose \solveorigin{\vu,V,Inv}. We will show that there exists a $\must\subseteq \X$
satisfying the weakened rules (W1-W6).
Fix a $\vu$. From Proposition~\ref{prop:complete}, the operator $\npst{}{\cdot}\P$ is upper bounded by 
$\pst{}{\cdot}$.
Since  $\pst{}{\cdot}$ has a fixed point $Inv$,
the fixed point of
$\npst{}{S}\P$ exists.
Let $\must$ be  the fixed point defined by W1-W2 and we have $\must \subseteq Inv$.
It follows that $\must\subseteq Inv\subseteq \safe$ and W3 holds.
W4 is inherited from  R4.
For any $C\subseteq \must \wedge C'\subseteq \npst{}{S}\P(C)$,
we have $\may\subseteq Inv$ and $\npst{}{S}\P(C)\subseteq post_\vu(C)$.
From R5, therefore,  $V(C)\geq V(C')$ and W5 holds.
Similarly, $\npst{}{S}\P(C,k)\subseteq post_\vu(C,k)$, thus W6 also holds.
Therefore, \solveunder{\vu,V,\must}.

The proof of second part is similar. 
\end{proof}

The above lemma states that the weakening and strengthening of the synthesis rules are sound. With the additional robustness condition, we can show that either the former is unsatisfiable or the latter is satisfiable. 

\begin{lemma}
\label{lem:complete}
If a $\synprob$ specified by $M$ is robust modulo $\C,\V$, then there exists a sufficiently fine partition $\P$ such that either 
\begin{inparaenum}[(i)]
\item $\Pi^w_\P(M,\C,\V)$ is unsatisfiable or 
\item $\Pi^s_\P(M,\C,\V)$ is satisfiable.
\end{inparaenum}
\end{lemma}
\begin{proof}
We discuss the two cases in Definition~\ref{def:robust}.
In this prove $post, \P-\overline{post},\P-\underline{post} $ without a subscript
denote the operator with respect to model $M$. 

Suppose there exists $\epsilon>0$ such that some controller $\vu$ and ranking function $V\in\V$ solves all $\epsilon$-perturbations of $\Pi(M,\C,\V)$. 
That is, for each $M'\in \B_\epsilon(M)$, there exists a $Inv_{M'}$,
such that $\vu,V,Inv_{M'}\models \Pi(M')$.
We define $Inv_\epsilon$ as the union of all such $Inv_{M'}$'s. 
Roughly, $Inv_\epsilon$ is the set of states that can be visited for some $\epsilon$-perturbation of $M$ with controller $\vu$. 
Since every $Inv_{M'}$ satisfies R3, R5, R6 in Figure~\ref{fig:rules},
it can be shown that (i) $Inv_\epsilon\subseteq \safe$, (ii) $C\subseteq Inv_\epsilon \wedge \pst{}{C}\cap C' \neq \emptyset \implies V(C) \geq V(C')$, and (iii) $C\subseteq Inv_\epsilon\backslash \goal  \wedge \pst{k}{C} \cap C'\neq \emptyset  \implies V(C)>V(C')  $.
Also, any subset of $Inv_\epsilon$ also satisfy the above three formula.
From Proposition~\ref{prop:complete}, for sufficiently fine partition $\P$, 
for any $S\subseteq \X$, 
$d(\wpst{}{S}\P,\pst{}{S})\leq \epsilon$. 
We will inductively prove that the $\may$ set with respect to this partition $\P$ is a subset of $Inv_\epsilon$.
(i) Initially, $\init\subseteq Inv_\epsilon$.
(ii) For any set $S\subseteq Inv_\epsilon$,
and any state $x\in\wpst{}{S}\P$, it suffice to prove $x\in Inv_\epsilon$. 
First, since $d(\wpst{}{S}\P,\pst{}{S})\leq \epsilon$.
We can find a state $x'\in \pst{}{S} \subseteq Inv_\epsilon$
such that $||x-x'||\leq \epsilon$. 
Since $x'$ is in $Inv_\epsilon$, it is reached by some model $M'\in\B_\epsilon(M)$ for the first time.
We construct a model $M''$ that is identical to $M'$ elsewhere except that
at state $x'$ the dynamics is $x' = post_{\M''}(x,\vu)$. It is easy to show that $M''$ is a $\epsilon$-perturbation of $M$ which visits $x$ with controller $\vu$.
Thus $x\in  Inv_\epsilon$. 
By (i) and (ii) above, we derive $\may\subseteq Inv_\epsilon$.
It follows that  the strengthened rules S3, S5 and S6 are satisfied.
In addition, S1-S2 is satisfied by the definition of $\must$ and S4 is just inherited from R4.
Therefore, the strengthened rules are satisfiable.

Otherwise suppose exists $\epsilon>0$ such that none of the $\epsilon$-perturbation of $\Pi(M,\C,\V)$ is satisfiable. 
Again from Proposition~\ref{prop:complete}, for sufficiently fine partition $\P$, 
for any $S\subseteq \X$, 
$d(\npst{}{S}\P,\pst{}{S})\leq \epsilon$. 
We on the contrary assume there exists some controller \solveunder{\vu,V,\must}.
We define a model $M'$ such that
for any cell $C\in \C$ and each state$x\in C$, the dynamics of $M'$ 
is captured by $post_{\M'}(x,\vu)=Proj(\pst{}{x},{\npst{}{C}\P})$. 
The operator $Proj(x,A)$ is a projection that maps $x$ to a state in $A$ that is closest to $x$.
It can be shown that $M'$ is an $\epsilon$-perturbation of $M$.
Moreover, for any cell $C\subseteq \X$, $post_{M'}{}{C} = \npst{}{C}\P$.
Thus, the problem of $\Pi(M'',\C,\V)$ and $\Pi^w_\P(M,\C,\V)$ are identical.
It follows that $\vu,V,\must \models \Pi(M',\C,\V)$, which contradicts the fact none of $\epsilon$-perturbation of $\Pi$ is satisfiable modulo $\C,\V$. 

\end{proof}

\section{SMT-based Synthesis Algorithm}
\label{sec:algorithm}


We introduce an algorithm for controller synthesis for $\synprob$ using the strengthened and weakened inductive SMT encodings of the previous section.
The algorithm takes as input the model $M$, 
the controller table/partition $\C$, the template for the ranking function $\V$
and the three sets $\init, \safe$ and $\goal$ that define $\synprob$ problem. 
It iteratively refines the partition $\P$ for representing invariants and makes subroutine calls to the SMT solver with the strengthened and weakened rules until it either finds a controller law $\vu$ or outputs $\bot$. 
Specifically, in each iteration, 
(a) if the strengthened problem $\Pi^s_\P$ is satisfiable then it returns the satisfying $\vu$.
(b) if the weakened problem $\Pi^w_\P$ is unsatisfiable then it returns $\bot$.
Otherwise, (c) it refines the partition $\P$ (using the $Must$ set computed from $\Pi^w_\P$).
\begin{algorithm}[h]
 {\bf input:\/}  $M,  \C, \V, \init, \safe, \goal$\;
  $\P \gets initPartition$\;
  \While{True}{
  $(\mathit{val}_s,\vu)\gets \mathit{Solve(\Pi^s_\P(M,\C,\V))}$\;
  $(\mathit{val}_w,\must)\gets \mathit{Solve(\Pi^w_\P(M,\C,\V))}$\;
  \uIf{$\mathit{val}_s = \mathit{SAT}$}{
  \Return $\vu$
  }
  \uElseIf{$\mathit{val}_w = \mathit{UNSAT}$}{
  \Return $\bot$
  }
  \Else{
  $\P\gets \mathit{Refine}(\P,\must)$\;
  }
  }
 \caption{SMT-based Synthesis Algorithm}
\end{algorithm}
The $\mathit{Refine}(\P,\must)$ function  
creates a finer partition of $\P$. 
For the completeness result, we require that for any $\P$, by iteratively applying
$\mathit{Refine}$, the resolution of the resulting partition
can be made arbitrarily fine.
In Section~\ref{sec:refine}, we discuss several heuristics for refinement
that potentially improve the performance of the algorithm.

\subsection{Soundness and Relative Completeness}

We will next sketch the arguments for the correctness of the algorithm. 
Soundness of the algorithm implies that 
whenever it outputs  $\vu$, 
(i) that $\vu$ is a control law that solves the $\synprob$ problem, 
(ii) the $\may$ set  obtained from solving $\Pi^s_\P$ in the final iteration
is an  inductive proof certificates for safety with $\vu$, and 
(iii) the $V$ is a $k$-step inductive proof certificate for progress with $\vu$. 
And, whenever the algorithm outputs $\bot$ then there 
does not exists a controller $\vu$, a ranking function $V\in \V$ and an invariant $Inv\subseteq$
such that the above (i)-(iii) holds.

In addition, we show that the algorithm is relative complete.
That is, if $\synprob$ is robust modulo $\C,\V$, then the algorithm terminates
with one of the above answers.

\begin{theorem}
\label{thm:main}
The algorithm is sound and relatively complete.
\end{theorem}
\begin{proof}
{\em Soundness.} If the algorithm terminates and return $\vu$,
then for some partition $\P$, the SMT solver returns a satisfying solution $\vu$ with 
$\Pi_\P^s(M)$.
From Lemma~\ref{lem:sound}, $\vu$ solves the $\synprob$.
Otherwise if the algorithm terminates and returns $\bot$,
then for some partition $\P$, the SMT solver on $\Pi^w_\P(M)$ returns UNSAT.
From Lemma~\ref{lem:sound}, there is no control that solves the $\synprob$ problem modulo $\V$.

{\em Relative Completeness.}
Since $\Pi(M)$ is a robust $\synprob$ modulo $\V$,
from Lemma~\ref{lem:complete}, we know that for a sufficiently fine partition,
either $\Pi_\P^w(M)$ is  unsatisfiable or $\Pi_\P^s(M)$ is
are satisfiable. Thus the while-loop will terminate as the algorithm creates  fine enough partitions.


\end{proof}

\subsection{Guided Refinement}
\label{sec:refine}
There are different ways in which the refinement of the partition $\P$ can be implemented without compromising the soundness and the relative completeness guarantees. 
The naive strategy of subdividing every equivalence class in $\P$, increases the size of the SMT problems quickly.  As  our algorithm solves both the weakened and strengthened versions of the problem simultaneously, we can marshall extra information in performing refinement.
For example, when  the weakened rules return a possible control $\vu$ along with 
its proof $V,\must$, even though this controller $\vu$ cannot be proven (to be safe and progress making) with the  strengthened rules, it can provide useful information for guiding the refinement.

\begin{definition}
For a partition $\P$ and a set $S$ that is preserved by $\P$, $\P'$ is a {\em $S$-guided refinement} of $\P$ if $\P'$ is derived by refining the cells of $\P$ that are in $S$.
\end{definition}

One key observation is that, a $\X\backslash \must$-guided refinement helps in generating safety proofs (S3 and W3), while a $\must$-guided refinement can improve the precesion of  progress proofs (S5-S6 and W5-W6).
The following proposition formalizes part of this intuition and states that for given a controller $\vu$, refining the cells in $\must$ does not improve the precision of the fixed-point $\must,\may$ computed by rules S1-S2 and W1-W2.
\begin{prop}
\label{prop:refine}
For any control $\vu$, any set $\init$ and any partitions $\P$,
let $\must,\may$ be the fixed point of operator $\npst{}{\cdot}\P$ and $\wpst{}{\cdot}\P$ containing $\init$. 
Let $\P'$ be a $\must$-guided refinement of $\P'$
and $\must',\may'$  be the fixed point of $\npst{}{\cdot}{\P'}$ and $\wpst{}{\cdot}{\P'}$ containing $\init$. 
Then, $\must = \must'$ and $\may = \may'$.
\end{prop}
By above proposition, a $\must$-guided refinement provides no help in generating better safety proofs. However, from Proposition~\ref{prop:complete}, a finer partition $\P$ increase the precision of $\npst{}{C}{\P}$ and $\wpst{}{C}{\P}$.
Since the rules S5-S6 and W5-W6 involve computing  $\npst{}{C}{\P}$ and $\wpst{}{C}{\P}$
for cells in $\may$ and $\must$ respectively, a $\must$-guided refinement possibly increases the precision of these rules.
Based on the above observations, we can adopt to the following heuristics for refinement:
If the $\must$ set is close to the unsafe set, perform $\X\backslash\must$-guided refinement, otherwise perform $\must$-guided refinement.

%
%
%
%
%


\section{Prototype Implementation and Experiments} \label{sec:case}

We implemented the synthesis algorithm in Python using the the CVC4 SMT solver~\cite{barrett2011cvc4}.
In this section, we briefly report preliminary results on applying it to  a simple class of navigation problems. With this implementation, we were able to automatically synthesize correct controls (and their inductive proofs) for some configurations and proved impossibility  for others.

\paragraph*{Vehicle Navigation Problem}

We consider a reach-avoid problem for a vehicle that follows piecewise linear approximation of Dubin's dynamics. The system model has $4$ state variables $[x,y,v,\theta]^T$: position, velocity and heading angle of the vehicle.
It has input variables $[\alpha,\beta]^T$: the acceleration and the turning rate.
From the  continuous Dubin's vehicle model:
$
\dot x =  v\cos\theta,\
\dot y  =  v\sin\theta,\
\dot v = \alpha,\
\dot \theta = v\beta.  
$
we construct a switched linear model by partitioning the domain of $\theta$ and $v$
into $24$ locations, and for each location we compute an approximate linear dynamics. 
The result is a switched-linear model:
\begin{equation}
\label{eq:vehicledym}
x^+ = x + av + b, \ y^+ = y+ cv + d, \ v^+ = v+\alpha, \ \theta^+ = \theta + e\beta,
\end{equation}
where $a,b,c,d,e$ have different values in different locations.
The piecewise linearized model preserves some properties
of the original system. For example, the linearized model cannot turn in place: if the velocity is close to 0, the heading $\theta$ cannot change.
Moreover, the velocity is non-negative, which further restricts its maneuverability.
These properties give rise to interesting $\synprob{}$ problem instances where the system has no satisfying control law.

We allow  finitely many discrete input values and 
compute $\npst{}{C}{\P}$  and $\wpst{}{C}{\P}$ offline as follows:
For a given partition $\P$, and a cell $C\in\P$, we first identify a set of cells $\N(C)$ such that $C'\in \N(C)$ can visit some state in $C$ in one step.
Then, for each possible input $\vu$, we compute the one step reach set of $\pst{}{\N(C)}$ with the help from reachability tools such as~\cite{frehse2011spaceex,duggiralac2e2}.
Thus we just need to identify the input combinations such that $C$ is covered by or intersected with $\pst{}{\N(C)}$.


\paragraph*{Experimental Results}

We performed several experiments for the above class of problems using our prototype implementation. We search for a control policy as a look-up table, specified by a $\CPar{}$.
We utilize a $\CPar{}$ $\C$ with 768 cells in total. 
In Figure~\ref{fig:sat} and~\ref{fig:unsat}, the grids illustrate the projection of \CPar{} to $x,y$ coordinates.

We create a partition $\P$ by further partitioning each cell in $\C$ into 4 pieces, with which we construct both the weakened and the strengthened rules.
For some cases, we proved the impossibility of synthesis. We visualize such a case in Figure~\ref{fig:unsat}. 
While for other cases, we successfully synthesized a control policy. An example is illustrated in Figure~\ref{fig:sat}. The satisfying control policy is synthesized with an inductive proof, namely the $\may$ set and the ranking function $\V$. 

\final{In the constraints of this synthesis problem, there are 768 real-valued variables for control input in each cell, 3072 integer variables for values of the ranking function for each partition and 3072 boolean variables indicates whether a partition is reached. 
The weakened or strengthened inductive rules are encoded in roughly 7000 constraints.
The constraints are solved by CVC4~\cite{barrett2011cvc4} in 10 minutes.
}

\begin{figure}
\centering
\includegraphics[width=\textwidth]{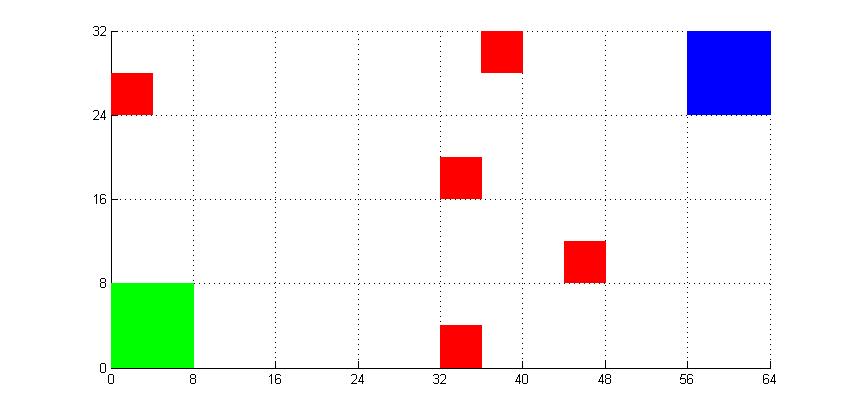}
\caption{A $\synprob{}$ instance that is impossible to solve. The grid illustrates the \CPar{}, the green block at the bottom left  corner is $\init$, the blue rectangle at the  top right is $\goal$, the smaller red blocks are unsafe. 
}
\label{fig:unsat}
\end{figure}
\begin{figure}

\centering
\includegraphics[width=\textwidth]{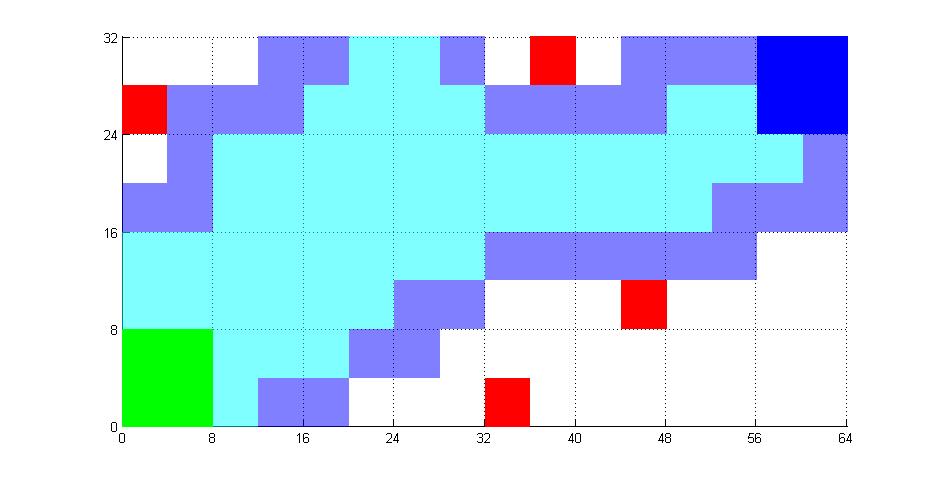}
\caption{A $\synprob{}$ instance that has a satisfying control law. 
The lighter connected region is the $\must$ set and the darker region together with the lighter region is the $\may$ set. 
}
\label{fig:sat}
\end{figure}

\section{Conclusion} 
\label{sec:conc}

In this work, we studied the controller synthesis problem of discrete-time systems with possibly unbounded time safety and progress specifications. 
Leveraging the growing strength of modern  SMT tools, we propose an algorithm that finds controllers as well as inductive proofs of their correctness.  Specifically, the algorithm creates a weaker and a stronger version of the synthesis problem and encodes them as SMT problems. By solving the controller synthesis problems for these two bounding systems automatically with SMT solvers, we can solve the synthesis problem for the original system. We prove that this algorithm is sound and relatively complete and show that the solution given by the strengthened system provide a guidance for refining the bounding system.
Our experimental results based on a prototype implementation suggest that this can be a promising direction of investigation for controller synthesis research. 

Since the core problem of computing over-approximations of $post$ are decoupled from synthesis in this formulation, one future direction of research that this work opens up is to  extend this framework to nonlinear system models. 
\final{The performance of the algorithm depends on the templates of the control, ranking function and invariants. 
Thus, to explore different classes of templates and study their performance in our synthesis framework is also a natural next step.}


\begin{thebibliography}{10}

\bibitem{aydin_gol_finite_2014}
E.~Aydin~Gol, Xuchu Ding, M.~Lazar, and C.~Belta.
\newblock Finite {Bisimulations} for {Switched} {Linear} {Systems}.
\newblock {\em IEEE Transactions on Automatic Control}, 59(12):3122--3134,
  December 2014.

\bibitem{barrett2011cvc4}
Clark Barrett, Christopher~L Conway, Morgan Deters, Liana Hadarean, Dejan
  Jovanovi{\'c}, Tim King, Andrew Reynolds, and Cesare Tinelli.
\newblock Cvc4.
\newblock In {\em Computer aided verification}, pages 171--177. Springer, 2011.

\bibitem{bemporad2002model}
Alberto Bemporad, Francesco Borrelli, Manfred Morari, et~al.
\newblock Model predictive control based on linear programming\~{} the explicit
  solution.
\newblock {\em IEEE Transactions on Automatic Control}, 47(12):1974--1985,
  2002.

\bibitem{beyene2014constraint}
Tewodros Beyene, Swarat Chaudhuri, Corneliu Popeea, and Andrey Rybalchenko.
\newblock A constraint-based approach to solving games on infinite graphs.
\newblock In {\em Proceedings of the 41st annual ACM SIGPLAN-SIGACT symposium
  on Principles of programming languages}, pages 221--234. ACM, 2014.

\bibitem{Bloem2012911}
Roderick Bloem, Barbara Jobstmann, Nir Piterman, Amir Pnueli, and Yaniv Sa'ar.
\newblock Synthesis of reactive(1) designs.
\newblock {\em Journal of Computer and System Sciences}, 78(3):911 -- 938,
  2012.
\newblock In Commemoration of Amir Pnueli.

\bibitem{CDP97}
Egon B\"{o}rger, Erich Gr\"{a}del, and Yuri Gurevich.
\newblock {\em The classical decision problem}.
\newblock Perspectives in mathematical logic. Springer, Berlin, Heidelberg, New
  York, 1997.

\bibitem{cardenas2009challenges}
Alvaro Cardenas, Saurabh Amin, Bruno Sinopoli, Annarita Giani, Adrian Perrig,
  and Shankar Sastry.
\newblock Challenges for securing cyber physical systems.
\newblock In {\em Workshop on future directions in cyber-physical systems
  security}, 2009.

\bibitem{cardenas2008research}
Alvaro~A C{\'a}rdenas, Saurabh Amin, and Shankar Sastry.
\newblock Research challenges for the security of control systems.
\newblock In {\em HotSec}, 2008.

\bibitem{chen2006}
Feng Chen and Grigore Rosu.
\newblock Mop: Reliable software development using abstract aspects.
\newblock {\em Technical Report UIUCDCS-R-2006- 2776, Department of Computer
  Science, University of Illinois at Urbana-Champaign}, October 2006.

\bibitem{de2008z3}
Leonardo De~Moura and Nikolaj Bj{\o}rner.
\newblock Z3: An efficient smt solver.
\newblock In {\em Tools and Algorithms for the Construction and Analysis of
  Systems}, pages 337--340. Springer, 2008.

\bibitem{ding2011reachability}
Jerry Ding, Eugene Li, Haomiao Huang, and Claire~J Tomlin.
\newblock Reachability-based synthesis of feedback policies for motion planning
  under bounded disturbances.
\newblock In {\em Robotics and Automation (ICRA), 2011 IEEE International
  Conference on}, pages 2160--2165. IEEE, 2011.

\bibitem{duggiralac2e2}
Parasara~Sridhar Duggirala, Sayan Mitra, Mahesh Viswanathan, and Matthew Potok.
\newblock C2e2: A verification tool for stateflow models.

\bibitem{dutertre2006yices}
Bruno Dutertre and Leonardo De~Moura.
\newblock The yices smt solver.
\newblock {\em Tool paper at http://yices. csl. sri. com/tool-paper. pdf},
  2(2), 2006.

\bibitem{frehse2011spaceex}
Goran Frehse, Colas Le~Guernic, Alexandre Donz{\'e}, Scott Cotton, Rajarshi
  Ray, Olivier Lebeltel, Rodolfo Ripado, Antoine Girard, Thao Dang, and Oded
  Maler.
\newblock Spaceex: Scalable verification of hybrid systems.
\newblock In {\em Computer Aided Verification}, pages 379--395. Springer, 2011.

\bibitem{gao_dreal:_2013}
Sicun Gao, Soonho Kong, and Edmund~M. Clarke.
\newblock {dReal}: {An} {SMT} {Solver} for {Nonlinear} {Theories} over the
  {Reals}.
\newblock In Maria~Paola Bonacina, editor, {\em Automated {Deduction} –
  {CADE}-24}, number 7898 in Lecture {Notes} in {Computer} {Science}, pages
  208--214. Springer Berlin Heidelberg, 2013.

\bibitem{kress-gazit_temporal-logic-based_2009}
H.~Kress-Gazit, G.E. Fainekos, and G.J. Pappas.
\newblock Temporal-{Logic}-{Based} {Reactive} {Mission} and {Motion}
  {Planning}.
\newblock {\em IEEE Transactions on Robotics}, 25(6):1370--1381, December 2009.

\bibitem{lahijanian_temporal_2012}
M.~Lahijanian, S.B. Andersson, and C.~Belta.
\newblock Temporal {Logic} {Motion} {Planning} and {Control} {With}
  {Probabilistic} {Satisfaction} {Guarantees}.
\newblock {\em IEEE Transactions on Robotics}, 28(2):396--409, April 2012.

\bibitem{liu_abstraction_2014}
Jun Liu and Necmiye Ozay.
\newblock Abstraction, {Discretization}, and {Robustness} in {Temporal} {Logic}
  {Control} of {Dynamical} {Systems}.
\newblock In {\em Proceedings of the 17th {International} {Conference} on
  {Hybrid} {Systems}: {Computation} and {Control}}, {HSCC} '14, pages 293--302,
  New York, NY, USA, 2014. ACM.

\bibitem{mazo2010pessoa}
Manuel Mazo~Jr, Anna Davitian, and Paulo Tabuada.
\newblock Pessoa: A tool for embedded controller synthesis.
\newblock In {\em Computer Aided Verification}, pages 566--569. Springer, 2010.

\bibitem{nedunuri_smt-based_2014}
S.~Nedunuri, S.~Prabhu, M.~Moll, S.~Chaudhuri, and L.E. Kavraki.
\newblock {SMT}-based synthesis of integrated task and motion plans from plan
  outlines.
\newblock In {\em 2014 {IEEE} {International} {Conference} on {Robotics} and
  {Automation} ({ICRA})}, pages 655--662, May 2014.

\bibitem{saha_automated_2014}
I.~Saha, R.~Ramaithitima, V.~Kumar, G.J. Pappas, and S.A. Seshia.
\newblock Automated composition of motion primitives for multi-robot systems
  from safe {LTL} specifications.
\newblock In {\em 2014 {IEEE}/{RSJ} {International} {Conference} on
  {Intelligent} {Robots} and {Systems} ({IROS} 2014)}, pages 1525--1532,
  September 2014.

\bibitem{svorenova_temporal_2015}
M{\'a}ria Svore{\v{n}}ov{\'a}, Martin Chmel{\'\i}k, Kevin Leahy, Hasan~Ferit
  Eniser, Krishnendu Chatterjee, Ivana {\v{C}}ern{\'a}, Calin Belta, et~al.
\newblock Temporal logic motion planning using pomdps with parity objectives.
\newblock 2015.

\bibitem{svorenova_temporal_2014}
Maria Svorenova, Jan Kretinsky, Martin Chmelik, Krishnendu Chatterjee, Ivana
  Cerna, and Calin Belta.
\newblock Temporal {Logic} {Control} for {Stochastic} {Linear} {Systems} using
  {Abstraction} {Refinement} of {Probabilistic} {Games}.
\newblock {\em arXiv:1410.5387 [cs]}, October 2014.
\newblock arXiv: 1410.5387.

\bibitem{turri2013linear}
Valerio Turri, Adriano Carvalho, H~Eric Tseng, Karl~H Johansson, and Francesco
  Borrelli.
\newblock Linear model predictive control for lane keeping and obstacle
  avoidance on low curvature roads.
\newblock In {\em Intelligent Transportation Systems-(ITSC), 2013 16th
  International IEEE Conference on}, pages 378--383. IEEE, 2013.

\bibitem{wongpiromsarn_receding_2012}
T.~Wongpiromsarn, U.~Topcu, and R.M. Murray.
\newblock Receding {Horizon} {Temporal} {Logic} {Planning}.
\newblock {\em IEEE Transactions on Automatic Control}, 57(11):2817--2830,
  November 2012.

\bibitem{zhou2012general}
Zhengyuan Zhou, Ryo Takei, Haomiao Huang, and Claire~J Tomlin.
\newblock A general, open-loop formulation for reach-avoid games.
\newblock In {\em CDC}, pages 6501--6506, 2012.

\end{thebibliography}
\end{document}